\documentclass[review]{elsarticle}

\usepackage{lineno,hyperref}
\usepackage{amsmath,array,amssymb,amsthm}
\usepackage[utf8]{inputenc}
\usepackage{amsfonts}
\usepackage{graphicx}
\usepackage{booktabs}
\usepackage{multirow}
\usepackage[a4paper, total={6.5in, 8in}]{geometry}
\usepackage{color}
\usepackage[inline]{enumitem}

\graphicspath{ {images/} }

\newtheorem{theorem}{Theorem}
\newtheorem{Coro}{Corollary}

\newcommand{\STAB}[1]{\begin{tabular}{@{}c@{}}#1\end{tabular}}

\newcommand{\E}{\mathbb {E}}
\newcommand{\Z}{\mathbb {Z}}

\newcommand{\Var}{\operatorname{Var}}

\newcommand{\betab}{\bff{\beta}}
\newcommand{\bff}[1]{\boldsymbol{#1}}

\newcommand{\Sigmab}{\bff{\Sigma}}

\newcommand{\eb}{\bff{e}}

\newcommand{\Wb}{\bff{W}}

\newcommand{\Sb}{\bff{S}}

\newcommand{\beq}{\begin{eqnarray}}
\newcommand{\eeq}{\end{eqnarray}}
\newcommand{\beqq}{\begin{eqnarray*}}
\newcommand{\eeqq}{\end{eqnarray*}}

\newcommand{\lp}{\left(}% abrir par{\^e}nteses
\newcommand{\rp}{\right)}% fechar par{\^e}nteses
\newcommand{\lc}{\left[}% abrir colchetes
\newcommand{\rc}{\right]}% fechar colchetes

\journal{Journal of Process Control}

\begin{document}

\begin{frontmatter}

\title{Fault detection and diagnosis of batch process using dynamic ARMA-based control charts \tnoteref{mytitlenote}}

\author{Batista Nunes de Oliveira, Marcio Valk}
\author{Danilo Marcondes Filho \corref{Danilo Marcondes Filho}}
\address{Department of Statistics, Federal University of Rio Grande do Sul, Av. Bento Gonçalves, 9500 – Porto Alegre, RS, Brazil}

\cortext[Danilo Marcondes Filho]{Corresponding author}
\ead{marcondes.filho@ufrgs.br}

\begin{abstract}
A wide range of approaches for batch processes monitoring can be found in the literature. This kind of process generates a very peculiar data structure, in which successive measurements of many process variables in each batch run are available. Traditional approaches do not take into account the time series nature of the data. The main reason is that the time series inference theory is not based on replications of time series, as it is in batch process data. It is based on the variability in a time domain. This fact demands some adaptations of this theory in order to accommodate the model coefficient estimates, considering jointly the batch to batch samples variability (batch domain) and the serial correlation in each batch (time domain). In order to address this issue, this paper proposes a new approach grounded in a group of control charts based on the classical ARMA model for monitoring and diagnostic of batch processes dynamics. The model coefficients are estimated (through the ordinary least square method) for each historical time series sample batch and modified \textit{Hotelling} and \textit{t-Student} distributions are derived and used to accommodate those estimates. A group of control charts based on that distributions are proposed for monitoring the new batches. Additionally, those groups of charts help to fault diagnosis, identifying the source of disturbances. Through simulated and real data we show that this approach seems to work well for both purposes.

\end{abstract}


\end{frontmatter}

%\linenumbers

\section{Introduction}

Industrial batch processes are used to produce a wide range of products. This kind of process generates for each batch run time series representing multiple measurements of process variables. That peculiar data structure containing samples of time series and data with a strong dynamic feature makes it still very challenging to develop control charts based on time series models. 

Traditional monitoring approaches do not take into account directly the time series nature of the data. Most of them decompose the tri-dimensional data array (batches $\times$ variables $\times$ time-instants) in a two-dimensional array (batches $\times$ variables/time-instants), based on the precursor approach of \linebreak \cite{Nomikos1995}. In this context, considering batches as sample replications, control approaches are proposed by using multivariate techniques (as Principal Components, Partial Least Squares, Discriminant Analysis, Support Vector machines, Neural Networks, etc) applied in the variable/time domain. These multivariate-based control charts are able to capture the dynamic data behavior in some way. We can mention a number of papers presenting improvements, applications and fruitful discussions in this direction in \cite{Wold1998}, \cite{jia2010}, \cite{van2012}, \cite{peng2014}, \cite{kuang2015}, \cite{wang2018}, \cite{peres2019}, \cite{li2020} and \cite{wang2021}. 

Control charts based on time series models are well known in the context of continuous processes. These kinds of processes have an intrinsic two-way data structure (samples $\times$ variables) since there are no replications of measures in each sample, i.e, there is no time dimension. Those models are mainly aimed to deal with the serial sample correlation, by using traditional control charts for the uncorrelated residuals or the cumulative-based charts for the fitted values. In both cases, as the first step, the model is adjusted from historical in-control samples and, in the next step, future samples are monitored through those charts. To accomplish that goal there are a wide range of propositions using the Autoregressive Moving Average (ARMA) and Vector Autoregressive (VAR) models. The good illustration can be seen in a case study presented in [Montgomery, 2007], in which a Shewhart control chart is used to monitor the residuals of the AR(1) model. In the context of a multivariate continuous process, the same goal is accomplished by using the VAR model for uncorrelated data and the Hotelling-based control chart for monitoring the vector of residuals. Some papers presenting the ground theory and additional contributions in such direction can be found in \cite{snoussi2011}, \cite{Pan2012}, \cite{Vanhatalo2015}, \cite{leoni2015} and \cite{simoes2016}. We can also find in the literature approaches for monitoring fitted values from those models. In this case, fitted values from new samples are monitored by using the \textit{exponentially weighted moving average} (EWMA) based control charts. A good review of those procedures can be seen in \cite{alwan1988}, \cite{jiang2000}, \cite{jiang2001}, \cite{jiang2007}, \cite{de2018}, \cite{de2019}, \cite{costa2021}, \cite{sheikhrabori2021}, \cite{nguyen2021} and \cite{jafarian2021}.

There are a few works in the literature presenting approaches for batch process monitoring based on time series models. The reason why this topic indeed hasn't been fully explored yet is that the time series inference theory is not based on replications of time series (each one bringing successive measurements of processes variables), as it is in batch process data. It is based on the variability in a time domain. This fact demands some adaptations of this theory in order to accommodate the coefficient estimates, taking into account the number of batch samples (batch to batch variability - batch domain analysis). The main problem is to build a single estimate for the model coefficients given a number of time series available, combining information from the batch and time domains. We highlight the work of \cite{Choi2008}, which propose a set of charts based on the traditional Hotelling statistic for the VAR residuals and fitted vector of observations, obtained thought the adjusted VAR from the historical time series samples batches in a reduced variable space. In this approach, the VAR coefficient estimates are done by using the Partial Least Square regression technique instead of by using the VAR estimation theory. \cite{wang2017} propose a group of control charts based on the 2D ARMA model. That model formulation try to capture the within-batch and batch-to-batch variability. For each batch they use the \textit{iterative step-wise regressions} (SWR) and the \textit{least absolute shrinkage and selection operator} (LASSO) to identify the model order and select the model coefficients. Future batch samples are monitored through control charts based in two stability index built from the coefficients estimates (the batch-to-batch and the within-batch index).  \cite{marcondes2020} presented a recent approach to deal with batch processes using VAR models focused on the VAR coefficients directly. In short, the VAR coefficients are estimated for each historical time series sample batch and by using a single estimate [as a combination of individual ordinary least square (OLS) estimates from each batch], the \textit{Hotelling} and the \textit{Generalized Variance} control charts are used for monitoring new batches. 

This paper proposes a new approach grounded in a group of control charts based on the classical ARMA model for dynamic monitoring and diagnostic of batch processes. Considering one variable at the time, we present a control approach based on the ARMA coefficients. The model coefficients are estimated for each historical time series sample batch and the control charts based on the modified \textit{Hotelling} distribution are used to monitor new batches. Additionally, we extend this idea by using a group of control charts, one for each ARMA coefficient, based on the modified \textit{t-Stutent} distribution, in order to help the diagnostic of disturbances detected by the \textit{Hotelling} chart. There are two meaningful contributions in our proposition. The first one is that through the modified \textit{Hotelling} and \textit{t-Stutent} distributions the model coefficient estimates generated from the number of historical batches can be easily accommodated. There is no need to use any estimation method based on complex algorithms and so deal with convergence and computational time issues. Also, the derived exact distributions makes the control charts capable to detect disturbances of any level, even in a scenario in which there are few in-control batch samples available. The second one is that, unlike these mentioned approaches, we address the fault diagnose problem by using a group of \textit{t-Stutent} control charts. We show through a simulated batch process that the proposed approach outperforms a competitor based on the model residuals (the most common approach used in the continuous and batch processes) and it is powerful in terms of disturbance diagnosis. Furthermore, this approach seems to work well when applied in a real data set.

In order to describe our proposition, the paper is organized as follows: Section 2 brings a detailed description of our methodology, including the basis of ARMA models, the \textit{Hotelling} and \textit{t-Stutent} modified statistics and the ARMA-based control charts. In Section 3, the proposed approach is illustrated through simulated batch data. Section 4 shows an application in a real data set. Conclusions are presented in Section 5.

\section{ARMA-based control approach}
\subsection{ARMA model}

The Autoregressive-Moving-Average (ARMA) model is widely used in time series analysis and forecasting due to the flexibility and suitable statistical properties. The class of ARMA model was popularized by \cite{box1970} and is characterized by a simple and  parsimonious formulation. It combines the Autoregressive (AR) model, which involves regressing a variable on its own lagged values, with moving averages (MA) model, which considers the error term as a linear combination of its own lagged terms. In general, we can write the ARMA model as follows:
\begin{equation} \label{eq:ARMAequation}
x_t=\phi_0+\phi_{1}x_{t-1}+\cdots+\phi_{v}x_{t-v}+\epsilon_{t}+\theta_{1}\epsilon_{t-1}+\cdots+\theta_{w}\epsilon_{t-w}, \quad t \in \Z,
\end{equation} 
where the error term is a white noise process (WN) with zero mean and variance $\sigma^2$, noted by ${\epsilon_t}\sim WN(0,\sigma^2), \, t\in \Z$.
This formulation is usually referred as ARMA($v,w$) since $v$ lags of the return are used as well as $w$ lags of the error term to specify the linear functional form to be estimated. Let's consider the vector of parameters
\begin{equation}\label{eq:beta}
\betab=[\phi_0,\phi_1,\dots,\phi_v,\theta_1,\dots, \theta_w].
\end{equation} 
Assuming that the process $\{x_t\}$ defined in \eqref{eq:ARMAequation} is causal and invertible, than for time-series of length $T$ sampled from this process, the asymptotic distribution (in $T$) of the OLS estimators $\widehat\betab$, where  $\widehat\betab=[\widehat\phi_0,\widehat\phi_1,\dots,\widehat\phi_v,\widehat\theta_0,\widehat\theta_1,\dots,\widehat\theta_w]$, is given by Theorem 8.11.1 in \cite{brockwell1991}. Under suitable conditions, it follows that

\beq\label{eq:multivariate}
(\widehat{\betab}-\betab)\stackrel{.}{\sim}N_p(\bf{0},\Sigmab_{\betab}),
\eeq
where $\Sigmab_{\betab}$ plays the role of the variance and covariance matrix of the $\betab$ estimators and  $\stackrel{.}{\sim}$  means asymptotic convergence in distribution, when $T$ increases.

Unfolding \eqref{eq:multivariate} we can write the univariate asymptotic distribution of each individual element of the vector $\widehat\betab$ in \eqref{eq:beta}. Let $\widehat\beta_*$ be any element of vector $\widehat\betab$, than we have

\beq\label{eq:univariate}(\widehat{\beta}_*-\beta_*)\stackrel{.}{\sim}N_p(0,\sigma^{2}_{\beta_*}),\eeq
where $\sigma^{2}_{\beta_*}$ is the corresponding element of the main diagonal of the $\Sigmab_{\betab}$.

In case of data coming from normally distributed variable, we consider the exact distribution of OLS estimated coefficients in \eqref{eq:multivariate} and \eqref{eq:univariate} rather than asymptotic one.
 
\subsection{Hottelling and t-Student adjusted distributions}

Following the aim of our work, we now assume the scenario with many trials available from an ARMA process, i.e., data samples representing time series. This is a typical context of a batch process that generates data representing trajectories of a variable to be considered under monitoring. In the next Section, the set of ARMA-based control charts for this kind of process will be proposed. They are based on the adaptation of the classical \textit{Hotelling} $\mathcal{T}^{2}$ statistic \cite{Montgomery2007}. In the two theorems below we demonstrate the distribution of the quantities that are the ground of our approach.

\begin{theorem}\label{theo1}

Let's consider $I$ time series of length $T$ from an ARMA($v,w$) process in \eqref{eq:ARMAequation}, $\betab$  the vector of parameters defined in \eqref{eq:beta} and $\widehat{\betab}_i$ the vector of OLS estimates for the $i^{st}$ sample, both vectors of length $p=v+w+1$. Assume that $\widehat\beta_{*i}$ is an element $\widehat{\betab}_i$. Than,

\beq\label{eq:betanormal}
\frac{\lp\widehat{\betab}_{*i}-\widehat{\overline{\betab}}_*\rp}{S_{\widehat{\beta}_{*i}}}\stackrel{.}{\sim} \sqrt{\frac{(I+1)}{I}}t_{(I-1)}, \,\,\mbox{as   }\, T \, \mbox{   increases},
\eeq
where $S^2_{\widehat{\beta}_{*i}}=\frac{1}{(I-1)}\sum_{i=1}^{I}\lp\widehat{\beta}_{*i}-\widehat{\overline{\beta}}_*\rp^2$, \, $\widehat{\overline{\beta}}_*=\frac{1}{I}\sum_{i=1}^{I}\widehat{\beta}_{*i}$ and $t_{(I-1)}$ is the $t-$Student distribution with $I-1$ degrees of freedom. 

\end{theorem}

\begin{proof}
Note that $\widehat{\beta}_{*1},\dots,\widehat{\beta}_{*I}$ are independent and identically distributed ($IID$) random variables with $\E(\widehat{\beta}_{*i})=\beta_{*}$ and  $\Var(\widehat{\beta}_{*i})=\sigma^2_{\beta_{*i}}$, for $i=1,\dots,I$. We must remember the following results of univariate distributions (\cite{mood1974}
):

\begin{enumerate}[label=(\roman*)]

\item If $\widehat{\beta}_{*}$ is normally distributed, than $\frac{(I-1)S^2_{\beta_{*i}}}{\sigma^2_{\beta_{*i}}} \sim \chi^2_{(I-1)}$.

\item Let $X$ and $Y$ be independent random variables, where $X\sim N(0,1)$ and $Y \sim \chi^2_{(\nu)}$, than $\frac{X}{\sqrt{\frac{Y}{v}}} \sim t_{(\nu)}$. 
\end{enumerate}

 Now, we can write

\beqq
\label{eq:univariates}
\frac{\lp\widehat{\betab}_{*i}-\widehat{\overline{\betab}}_*\rp}{S_{\widehat{\beta}_{*i}}}
&=&\frac{\frac{\lp\widehat{\betab}_{*i}-\widehat{\overline{\betab}}_*\rp}{\sigma_{\beta_{*i}}}} {\sqrt{\frac{S^2_{\widehat{\beta}_{*i}}}{\sigma^2_{\beta_{*i}}}}}= \frac{\Delta}{\sqrt{\Gamma}}.    
\eeqq

The quantities $\Delta$ and $\Gamma$ can be unfolded like:

\beqq
\Delta&=& \frac {\lp\widehat{\betab}_{*i}-\betab_*\rp}  {\sigma_{\beta_{*i}}} 
-\frac {\lp\widehat{\overline{\betab}}_* - \betab_*\rp}  {\sigma_{\beta_{*i}}} 
= \frac {\lp\widehat{\betab}_{*i}-\betab_*\rp} 
{\sigma_{\beta_{*i}}} 
-\frac{1}{\sqrt{I}}\frac {\lp\widehat{\overline{\betab}}_* - \betab_*\rp}  {\frac{\sigma_{\beta_{*i}}}{\sqrt{I}}}\nonumber\\ &=& C-\frac{1}{\sqrt{I}}D.
\eeqq

The variables $C$ and $D$ are independent and, for large $T$,  $C\stackrel{.}{\sim} N(0,1)$ and $D\stackrel{.}{\sim} N(0,1)$. Consequently, $\Delta \stackrel{.}{\sim} N\lp 0,\frac{I+1}{I}\rp$ or $\sqrt{(\frac{I}{I+1})}\Delta \stackrel{.}{\sim} N(0,1)$. From (i),

\beqq
\Gamma= {\frac{S^2_{\widehat{\beta}_{*i}}}{\sigma^2_{\beta_{*i}}}} 
= {\frac{(I-1)S^2_{\widehat{\beta}_{*i}}} {\frac{\sigma^2_{\beta_{*i}}}{(I-1)}}}
\stackrel{.}{\sim} \frac{\chi^2_{(I-1)}}{I-1}.
\eeqq

By $\mathrm{(ii)}$ it follows that, 

$$\frac{\Delta}{\sqrt{\Gamma}}\stackrel{.}{\sim}\sqrt{\frac{(I+1)}{I}}t_{(I-1)}, \,\,\mbox{as   }\, T \, \mbox{   increases}.$$
\end{proof}

\begin{Coro}\label{coro1}
 In case of data coming from a normally distribute variable, the distribution in \eqref{eq:betanormal} becomes an exact distribution rather than asymptotic one.
\end{Coro}

\begin{proof}
The same as in Theorem 1.
\end{proof}

\begin{theorem}\label{theo2}
Let's consider $I$ time series of length $T$ from an ARMA($v,w$) process in \eqref{eq:ARMAequation}, $\betab$  the vector of parameters defined in \eqref{eq:beta} and $\widehat{\betab}_i$ the vector of OLS estimates for the $i^{st}$ sample, both vectors of length $p=v+w+1$. Than,

\beq \label{eq:Teo2}
\lp\widehat{\betab}_i-\widehat{\overline{\betab}}\rp\Sb_{\widehat{\betab}}^{-1}\lp\widehat{\betab}_i-\widehat{\overline{\betab}}\rp\stackrel{.}{\sim}\frac{(I-1)(I+1)p}{I(I-p)p}F_{p,I-p}, \,\,\mbox{as   }\, T \, \mbox{   increases},
\eeq
where \quad $\Sb_{\widehat{\betab}}=\frac{1}{(I-1)}\sum_{i=1}^{I}\lp\widehat{\betab}_i-\widehat{\overline{\betab}}\rp\lp\widehat{\betab}_i-\widehat{\overline{\betab}}\rp'$ \quad \mbox{and} \quad $\widehat{\overline{\betab}}=\frac{1}{I}\sum_{i=1}^{I}\widehat{\betab}_i$.
\end{theorem}

\begin{proof}
Note that $\widehat{\betab}_1,\dots,\widehat{\betab}_{I}$ are $IID$ random vectors with $\E(\widehat{\betab}_i)=\betab$ and  $\Var(\widehat{\betab}_i)=\Sigma_{\betab}$, for $i=1,\dots,I$.  We must remember the following results of multivariate distributions (\cite{johnson2002}):

 \begin{enumerate}[label=(\roman*)]

\item If $\widehat{\betab}_i$ is normally distributed, than  $(I-1)\Sb_{\widehat{\betab}}\sim \Wb_{p(I-1)}\Sigmab_{\betab}$ ($\Wb$ is the Wishart distribution); 
    
\item $\widehat{\betab}_i$ and $\Sb_{\widehat{\betab}}$ are independent; %\cite{Johnson and Wichern}.

\item It follows from $\mathrm{(i)}$ and $\mathrm{(ii)}$ that $\lp\widehat{\betab}_i-\betab\rp\Sb_{\widehat{\betab}}^{-1}\lp\widehat{\betab}_i-\betab\rp\sim \frac{(I-1)p}{I-p}F_{p,(I-1)}$.
\end{enumerate}

Now, let's find the distribution of $(\widehat{\betab}_i-\widehat{\overline{\betab}})$:

\beqq
\widehat{\betab}_i-\widehat{\overline{\betab}}&=& \lp\widehat{\betab}_i-\betab\rp-\left(\widehat{\overline{\betab}}-\betab \right) 
= \lp\widehat{\betab}_i-\betab\rp-\lp\frac{1}{I}\sum_{i=1}^{I}\widehat{\betab}_i-\betab\rp\nonumber\\
&=&\lp\widehat{\betab}_i-\betab\rp-\frac{1}{I}\sum_{i=1}^{I}\lp\widehat{\betab}_i-\betab\rp.
\eeqq  
It follows from \eqref{eq:multivariate} that

\beqq
\lp\widehat{\betab}_i-\betab\rp \stackrel{.}{\sim} N_p(\bf{0},\Sigmab_{\betab})
\eeqq 
and

\beqq\frac{1}{I}\sum_{i=1}^{I}\left(\widehat{\betab}_i-\betab\right) \stackrel{.}{\sim} N_p\lp\textbf{0},\frac{1}{I}\Sigmab_{\betab}\rp.
\eeqq.

So, 
\beqq
\lp\widehat{\betab}_i-\widehat{\overline{\betab}}\rp \stackrel{.}{\sim}  N_p\lp\textbf{0},\lc\frac{I+1}{I}\rc\Sigmab_{\betab}\rp, 
\eeqq
or

\beqq
\sqrt{\frac{I}{I+1}}\lp\widehat{\betab}_i-\widehat{\overline{\betab}}\rp \stackrel{.}{\sim} N_p\lp\bf{0},\Sigmab_{\betab}\rp.
\eeqq
Rewriting $\mathrm{(iii)}$ explicitly in terms of probability distributions and considering \eqref{eq:multivariate}, we have

\beqq
N_p\lp\bf{0},\Sigmab_{\betab}\rp'\lp\frac{\Wb_{p(I-1)}\Sigmab_{\betab}}{I-1} \rp N_p\lp\bf{0},\Sigmab_{\betab}\rp \stackrel{.}\sim  \frac{(I-1)p}{I-p}F_{p,(I-1)}.
 \eeqq

Finally, from the equation above we can write $\mathrm{(iii)}$ like:

\beqq
\sqrt{\frac{I}{I+1}} N_p\lp\bf{0},\Sigmab_{\betab}\rp'\lp\frac{\Wb_{p(I-1)}\Sigmab_{\betab}}{I-1} \rp \sqrt{\frac{I}{I+1}} N_p\lp\bf{0},\Sigmab_{\betab}\rp,  
 \eeqq
or

\beqq
\frac{I}{I+1} N_p\lp\bf{0},\Sigmab_{\betab}\rp'\lp\frac{\Wb_{p(I-1)}\Sigmab_{\betab}}{I-1} \rp N_p\lp\bf{0},\Sigmab_{\betab}\rp.   
\eeqq

Now we can note that the quantity $\lp\widehat{\betab}_i-\widehat{\overline{\betab}}\rp\Sb_{\widehat{\betab}}^{-1}\lp\widehat{\betab}_i-\widehat{\overline{\betab}}\rp$ has the following probability distribution:

\beqq 
N_p\lp\bf{0},\Sigmab_{\betab}\rp'\lp\frac{\Wb_{p(I-1)}\Sigmab_{\betab}}{I-1} \rp N_p\lp\bf{0},\Sigmab_{\betab}\rp  \stackrel{.}\sim  \frac{(I+1)}{I} \frac{(I-1)p}{(I-p)}F_{p,(I-1)}, \,\,\mbox{as   }\, T \, \mbox{   increases}.
\eeqq
\end{proof}

\begin{Coro}\label{coro2}
 In case of data coming from a normally distributed variable, the distribution in \eqref{eq:Teo2} becomes an exact distribution rather than asymptotic one.
\end{Coro}

\begin{proof}
The same as in Theorem 2.
\end{proof}

\subsection{Dynamic ARMA-based control charts} 

 Consider a historical data set of $I$ batches yielding products compliant with specifications. For each batch we have a time series representing the trajectory of one variable, measured at $T$ time-instants, from the process under normal regime (in-control sample batches). Let's assume that the variable dynamics can be described by the ARMA process. 
   
 In order to find a reference distribution of the ARMA ($v,w$) coefficient estimates in Phase $\mathcal{I}$, we firstly save the OLS vector of estimates $\widehat{\betab}_i$ for each batch. Considering that $E(\widehat{\betab}_i)=\widehat{\betab}$ and $E(\Sb_{\widehat{\betab}})={\Sigmab}_{\betab}$, in the next step, we build the unique estimates of the mean and the covariance of $\widehat{\betab}_i$ by combining the individual estimates like:
  
  \beq \widehat{\overline{\betab}}= \frac{1}{I} \sum_{i=1}^{I} \widehat{\betab}_{i} \quad \mbox{and} \quad \Sb_{\widehat{\betab}}=\frac{1}{(I-1)}\sum_{i=1}^{I}(\widehat{\betab}_i-\widehat{\overline{\betab}})(\widehat{\betab}_i-\widehat{\overline{\betab}})'. \eeq

 These estimates hold relevant information about the variable dynamic (serial correlation) of the process operating in a normal regime. In Phase $\mathcal{II}$ we propose one approach based on the modified \textit{Hotelling} $\mathcal{T}^{2}$ statistic to monitor the future batch samples. We have shown in Theorem \ref{theo2} that
 
\beq \label{eq:hot}
\mathcal{T}_{\betab}^{2}=(\widehat{\betab}_i-\widehat{\overline{\betab}})\Sb_{\widehat{\betab}}^{-1}(\widehat{\betab}_i-\widehat{\overline{\betab}})\stackrel{.}{\sim}\frac{(I-1)(I+1)p}{I(I-p)p}F_{p,I-p},
\eeq
where $p=v+w+1$. Scores above the $\alpha$ percentile in $\mathcal{T}_{\betab}^2$ imply that the variable dynamics in a new batch are different to their expected behaviour for the in-control process.
Once the $\mathcal{T}_{\betab}^{2}$ chart pointed out a batch sample out of limit, we can investigate the coefficients most affected by using the control chart based on the modified \textit{t-Student} distribution. We have shown in Theorem \ref{theo1} that

\beq\label{eq:stud}
t_{\betab}=\frac{(\widehat{\betab}_{*i}-\widehat{\overline{\betab}}_*)}{S_{\widehat{\beta}_{*i}}}\stackrel{.}{\sim} \sqrt{\frac{(I+1)}{I}}t_{(I-1)},
\eeq
where $\widehat{\betab}_{*i}$ and $\widehat{\overline{\betab}}_*$ are elements of the vectors  $\widehat{\betab}_{i}$ and $\widehat{\overline{\betab}}$, respectively. $S^2_{\widehat{\beta}_{*i}}$ is an element of the main diagonal of $\Sb_{\widehat{\betab}}$. Scores above the $\alpha$ percentile in $t_{\betab}$ can signalize a disturbance in the dynamic caused by the changing in a specific coefficient.

\section{Simulation study}

In this Section, we generate batch processes in which the dynamic is described by an ARMA($v$,$w$) model.  In order to illustrate our method, we present a Monte Carlo simulation using varieties of this model, including combinations of $v,w=0,1,2$. Both models with intercept term. The model with the highest number of parameters in this study is an ARMA(2,2), explicitly written as 
 
\beq 
x_t=\phi_0+\phi_{1}x_{t-1}+\phi_{2}x_{t-2}+\epsilon_{t}+\theta_{1}\epsilon_{t-1}+\theta_{2}\epsilon_{t-2}
\eeq
with the vector of parameters $\betab=[\phi_0,\phi_1,\phi_2,\theta_1, \theta_2]$. Table \ref{tab:tbnew} show the set of ARMA parameters for in-control process and the simulation settings. In phase $\mathcal{II}$, we considering scenarios with a wide range of disturbances in the intercept term $\phi_{0}$ and in the AR/MA part of the model, represented by $\phi_{1}$ and $\theta_{1}$, respectively. 

\begin{table}[h!]\caption{Simulation settings }\label{tab:tbnew}
\begin{center}
\scriptsize
\begin{tabular}{|c|c|c|c|c|c|c|c|}
\hline
ARMA & ARMA & Disturbed & Disturbance& \# Batches &\# Batches & Batch & \multirow{2}{*
}{Run}\\
coefficients & settings  & parameter &  levels &phase $\mathcal{I}$ & phase $\mathcal{II}$& length $T$\\
\hline
% \cline{2-10}
$\phi_{0}$ & 1 &  X & $0,0.5,0.8,1,1.2,1.5,2$& & && \\
\cline{1-4}
$\phi_{1}$ & 0.2 &X & $-0.2,0,0.1,0.2,0.3$& & & &\\
\cline{1-4}
$\phi_{2}$ & 0.5 & & & 30, 50,  & 500 & 100, 200,  & 1000\\
\cline{1-4}
$\theta_{1}$ & 0.5 &X & $0,0.3,0.4,0.5,0.6,0.8$ & 100& &500, 1000&\\
\cline{1-4}
$\theta_{2}$ & -0.3& & & & && \\

\hline
\end{tabular}
\end{center}
\end{table}   
  
 We generate scenarios including different numbers of batches with different time-length (from 100 to 1000). Each scenario was replicated 1000 times. In phase $\mathcal{I}$ we do variate the number of batches from 30 to 100. The $\mathcal{T}_{\betab}^2$ \eqref{eq:hot} and $t_{\betab}$ \eqref{eq:stud} charts were setting to the false alarm probability of $\alpha=0.01$. In phase $\mathcal{II}$ 500 batches were generated in each scenario. The rate of batches beyond the control ($r$) and the $ARL$ index (\textit{Average Run Length}) were adopted to evaluate the chart's performance, where $ARL=1/r$. The  $ARL_{0}$ is the average number of batches until a false alarm (for $\alpha=1\%$, $ARL_{0}=100$), i.e., points above control limits in the process without disturbances (in-control process). In contrast, $ARL_{1}$ is the average number of samples until an out-of-control batch falls outside the control limits. The former is a measure of the chart's sensibility.

As a benchmark approach, we consider the usual way to build time series-based control charts for monitoring the residuals from the fitted model in phase $\mathcal{I}$ \cite{Montgomery2007} by adapting this methodology for the case of batch processes. We use as the variable the residual mean $\overline{e}$ from the $T-p$ residuals for each batch, where $p=v+w+1$ for the ARMA($v,w$) model with intercept. We know that if a new batch comes from the in-control process, $\overline{e}\stackrel{.}{\sim}N(0,1/\sqrt{T-p})$. We set the limits of the $t_{\eb}$ residual control chart using the probability of false alarm of $\alpha=1\%$. The EWMA approach \cite{Montgomery2007} is used in order to improve the power of the $t_{\eb}$ chart to detect disturbances representing small changes in the residual mean. Simulations and calculations were conducted using $R$ \cite{R}. 

Tables \ref{Tab2} to \ref{Tab4}  summarize the results of $\mathcal{T}_{\betab}^2$ and $t_{\eb}$ charts for an ARMA(1,1) model with the in-control parameters set in Table \ref{tab:tbnew}. The tables show the mean and standard deviation of $ARL$ values for each disturbance. The scenarios in Tables \ref{Tab2} and \ref{Tab3} are very similar in terms of results and so they will not be commented apart. These Tables include disturbances in AR coefficient  $\phi_{1}$ and MA coefficient $\theta_{1}$, each one at the time. The observed $ARL_{0}$ is close to the chosen nominal value of $100$, which is consistent with the fact that no disturbance was introduced in the process. The  $ARL_{1}$ values show that the $\mathcal{T}_{\betab}^2$ chart outperforms the $t_{\eb}$ in detecting disturbances of different intensities. Additionally, we notice that the degree of detection in $\mathcal{T}_{\betab}^2$ chart increases faster as the perturbations get more intense. Even for the higher values of disturbances, the performance of our approach remains better than the residual-based charts. We emphasize here the power of the proposed approach (based on estimates of correlations in $\phi$) to capture information about process dynamics.

\begin{table}[!ht] 
\caption{ARMA(1,1): Mean  $(\hat{\mu})$ and standard deviation  $(\hat{\sigma})$ of
 $ARL_0$ ($\phi_{1}=0.2$) and $ARL_1$ values for disturbances in the AR coefficient $\phi_{1}$}
 \label{Tab2}
\scriptsize
\centering
\begin{tabular}{c|c|cccc|cccc|cccc}
  \toprule
      \multirow{3}{*}{$\phi_1$} & \multirow{3}{*}{$n$} & \multicolumn{12}{c}{I} \\
\cline{3-14}
    &&\multicolumn{4}{c|}{10} &
     \multicolumn{4}{c|}{30 } &
      \multicolumn{4}{c}{100 } \\

\cline{3-14}
  &  & $\mathcal{T}_{\betab}^2(\hat{\mu})$ & $\mathcal{T}_{\betab}^2(\hat{\sigma})$ & $t_{\eb}(\hat{\mu})$&
  $t_{\eb}(\hat{\sigma})$&$\mathcal{T}_{\betab}^2(\hat{\mu})$ & $\mathcal{T}_{\betab}^2(\hat{\sigma})$ & $t_{\eb}(\hat{\mu})$& $t_{\eb}(\hat{\sigma})$& $\mathcal{T}_{\betab}^2(\hat{\mu})$ & $\mathcal{T}_{\betab}^2(\hat{\sigma})$ & $t_{\eb}(\hat{\mu})$& $t_{\eb}(\hat{\sigma})$\\
  \hline
     & 100 & 1.55 & 0.41 & 207.10 & 198.22 & 1.37 & 0.23 & 261.99 & 206.90 & 1.32 & 0.12 & 301.52 & 195.96 \\                 
-0.2 & 200 & 1.03 & 0.04 & 108.04 & 136.28 & 1.02 & 0.02 & 245.43 & 190.74 & 1.01 & 0.01 & 340.91 & 159.01 \\                 
     & 500 & 1.00 & 0.00 & 206.38 & 191.31 & 1.00 & 0.00 & 213.12 & 176.74 & 1.00 & 0.00 & 272.73 & 154.07 \\                 
     &1000 & 1.00 & 0.00 & 235.10 & 197.34 & 1.00 & 0.00 & 244.09 & 199.38 & 1.00 & 0.00 & 315.15 & 186.85 \\                 
\hline                                                                                                                       
     & 100 & 10.47 & 7.85 & 198.35 & 184.84 & 7.69 & 3.47 & 223.94 & 183.48 & 6.67 & 2.35 & 242.15 & 184.06 \\                
0.0  & 200 & 3.88 & 2.66 & 158.03 & 164.89 & 2.89 & 0.97 & 224.23 & 184.59 & 2.64 & 0.59 & 317.85 & 174.84 \\                 
     & 500 & 1.24 & 0.20 & 142.02 & 161.94 & 1.17 & 0.11 & 208.81 & 167.39 & 1.11 & 0.05 & 320.03 & 192.10 \\                 
     &1000 & 1.00 & 0.01 & 160.56 & 168.54 & 1.00 & 0.01 & 210.58 & 188.11 & 1.00 & 0.00 & 292.56 & 183.92 \\                
\hline                                                                                                                       
     & 100 & 50.61 & 53.17 & 163.09 & 170.65 & 30.24 & 19.69 & 169.52 & 151.33 & 29.49 & 20.40 & 223.92 & 177.65 \\           
 0.1 & 200 & 37.62 & 39.76 & 136.60 & 136.04 & 19.03 & 11.79 & 176.14 & 162.34 & 16.92 & 8.64 & 219.03 & 164.46 \\            
     & 500 & 8.31 & 6.87 & 158.23 & 169.50 & 7.07 & 7.29 & 164.81 & 150.12 & 4.75 & 1.42 & 268.94 & 183.23 \\                 
     &1000 & 2.88 & 1.67 & 193.80 & 193.80 & 2.19 & 0.86 & 185.32 & 171.76 & 1.92 & 0.32 & 236.66 & 170.04 \\               
      \hline                                                                                                                       
      & 100 & 143.14 & 149.41 & 116.17 & 144.96 & 114.40 & 107.31 & 132.28 & 153.58 & 106.96 & 97.14 & 182.32 & 168.63 \\      
 0.2 & 200 & 154.88 & 146.54 & 125.29 & 160.91 & 129.26 & 112.50 & 111.44 & 129.12 & 121.95 & 109.16 & 180.62 & 152.30 \\     
      & 500 & 198.93 & 167.21 & 116.80 & 150.20 & 158.91 & 145.74 & 171.53 & 174.05 & 132.75 & 116.87 & 179.84 & 150.86 \\     
      &1000 & 145.57 & 139.15 & 121.50 & 157.18 & 195.10 & 160.02 & 152.59 & 155.07 & 149.26 & 127.64 & 190.92 & 166.38 \\    
\hline                                                                                                                       
     & 100 & 89.42 & 123.93 & 74.85 & 113.05 & 64.15 & 86.76 & 104.44 & 125.94 & 36.17 & 21.05 & 88.44 & 111.81 \\            
0.3  & 200 & 65.90 & 110.76 & 66.32 & 100.31 & 32.98 & 52.26 & 101.62 & 120.87 & 20.54 & 10.19 & 104.85 & 129.28 \\           
     & 500 & 10.05 & 7.85 & 94.77 & 135.57 & 6.66 & 3.24 & 95.23 & 126.16 & 5.51 & 1.94 & 89.50 & 113.80 \\                   
     &1000 & 2.91 & 1.58 & 99.17 & 123.37 & 2.29 & 0.78 & 118.02 & 140.85 & 1.92 & 0.37 & 115.91 & 133.89 \\                 
\hline                                                                                                                       
     & 100 & 1.46 & 0.44 & 7.20 & 5.67 & 1.30 & 0.21 & 7.26 & 4.08 & 1.23 & 0.12 & 6.77 & 2.66 \\                             
0.6  & 200 & 1.02 & 0.03 & 6.68 & 4.59 & 1.01 & 0.01 & 6.98 & 3.81 & 1.00 & 0.00 & 6.72 & 2.42 \\                             
     & 500 & 1.00 & 0.00 & 5.76 & 2.33 & 1.00 & 0.00 & 6.63 & 2.90 & 1.00 & 0.00 & 7.03 & 2.68 \\                             
     &1000 & 1.00 & 0.00 & 7.52 & 6.95 & 1.00 & 0.00 & 8.04 & 10.00 & 1.00 & 0.00 & 7.46 & 3.39 \\          
\bottomrule
\end{tabular}
\end{table}

\begin{table}[!ht] 
\caption{ARMA(1,1): Mean  $(\hat{\mu})$ and standard deviation  $(\hat{\sigma})$ of
 $ARL_0$ ($\theta_{1}=0.5$) and $ARL_1$ values for disturbances in the MA coefficient $\theta_{1}$}
 \label{Tab3}
\scriptsize
\centering
\begin{tabular}{c|c|cccc|cccc|cccc}
  \toprule
      \multirow{3}{*}{$\theta_1$} & \multirow{3}{*}{$n$} & \multicolumn{12}{c}{I} \\
\cline{3-14}
    &&\multicolumn{4}{c|}{10} &
      \multicolumn{4}{c|}{30 } &
      \multicolumn{4}{c}{100 } \\

\cline{3-14}
  &  & $\mathcal{T}_{\betab}^2(\hat{\mu})$ & $\mathcal{T}_{\betab}^2(\hat{\sigma})$ & $t_{\eb}(\hat{\mu})$&
  $t_{\eb}(\hat{\sigma})$&$\mathcal{T}_{\betab}^2(\hat{\mu})$ & $\mathcal{T}_{\betab}^2(\hat{\sigma})$ & $t_{\eb}(\hat{\mu})$& $t_{\eb}(\hat{\sigma})$& $\mathcal{T}_{\betab}^2(\hat{\mu})$ & $\mathcal{T}_{\betab}^2(\hat{\sigma})$ & $t_{\eb}(\hat{\mu})$& $t_{\eb}(\hat{\sigma})$\\
\hline
     & 100  & 1.07 & 0.07 & 229.47 & 211.46 & 1.05 & 0.03 & 231.92 & 197.23 & 1.04 & 0.02 & 287.27 & 187.26 \\                
0.0  & 200  & 1.00 & 0.00 & 257.31 & 193.96 & 1.00 & 0.00 & 171.38 & 186.29 & 1.00 & 0.00 & 359.66 & 190.02 \\                
     & 500  & 1.00 & 0.00 & 218.69 & 189.65 & 1.00 & 0.00 & 276.63 & 208.32 & 1.00 & 0.00 & 272.92 & 149.85 \\                
     &1000  & 1.00 & 0.00 & 257.65 & 223.60 & 1.00 & 0.00 & 314.74 & 180.34 & 1.00 & 0.00 & 427.78 & 148.14 \\                
\hline                                                                                                                        
     & 100  & 6.23 & 3.60 & 161.44 & 165.48 & 5.21 & 1.95 & 201.47 & 177.20 & 4.65 & 1.25 & 241.61 & 184.57 \\                
0.3  & 200  & 2.58 & 1.13 & 153.24 & 181.85 & 2.25 & 0.67 & 216.95 & 194.58 & 2.00 & 0.32 & 316.91 & 188.82 \\                
     & 500  & 1.12 & 0.10 & 161.22 & 164.04 & 1.07 & 0.06 & 203.95 & 177.16 & 1.05 & 0.03 & 288.90 & 186.49 \\                
     &1000  & 1.00 & 0.00 & 176.50 & 174.51 & 1.00 & 0.00 & 217.00 & 189.18 & 1.00 & 0.00 & 289.73 & 188.91 \\                
\hline                                                                                                                        
     & 100  & 26.83 & 20.14 & 162.86 & 187.89 & 24.73 & 15.73 & 195.65 & 182.32 & 20.11 & 8.02 & 238.83 & 172.32 \\           
 0.4 & 200  & 17.30 & 15.87 & 148.62 & 176.75 & 13.48 & 8.29 & 192.74 & 178.25 & 11.14 & 5.23 & 254.78 & 191.96 \\            
     & 500  & 6.04 & 4.39 & 137.97 & 160.31 & 4.05 & 2.01 & 207.24 & 182.84 & 3.44 & 0.88 & 188.90 & 138.53 \\                
     &1000  & 2.03 & 0.87 & 145.61 & 165.84 & 1.67 & 0.33 & 181.28 & 171.34 & 1.53 & 0.22 & 219.33 & 167.85 \\                
\hline                                                                                                                        
     & 100  & 143.14 & 149.41 & 116.17 & 144.96 & 114.40 & 107.31 & 132.28 & 153.58 & 106.96 & 97.14 & 182.32 & 168.63 \\     
 0.5 & 200  & 154.88 & 146.54 & 125.29 & 160.91 & 129.26 & 112.50 & 111.44 & 129.12 & 121.95 & 109.16 & 180.62 & 152.30 \\    
    & 500  & 198.93 & 167.21 & 116.80 & 150.20 & 158.91 & 145.74 & 171.53 & 174.05 & 132.75 & 116.87 & 179.84 & 150.86 \\    
     &1000  & 145.57 & 139.15 & 121.50 & 157.18 & 195.10 & 160.02 & 152.59 & 155.07 & 149.26 & 127.64 & 190.92 & 166.38 \\    
\hline                                                                                                                        
     & 100  & 76.74 & 102.94 & 96.58 & 126.08 & 59.45 & 61.19 & 89.62 & 116.01 & 46.27 & 54.90 & 127.97 & 139.03 \\           
0.6  & 200  & 46.68 & 80.02 & 96.13 & 129.03 & 24.11 & 21.61 & 97.20 & 117.70 & 20.85 & 12.69 & 170.86 & 165.75 \\            
     & 500  & 11.66 & 49.70 & 99.45 & 127.41 & 4.95 & 2.69 & 117.78 & 129.75 & 4.05 & 1.21 & 158.33 & 154.58 \\               
     &1000  & 2.09 & 1.18 & 105.79 & 153.59 & 1.69 & 0.48 & 102.90 & 127.67 & 1.56 & 0.27 & 163.78 & 159.80 \\                
\hline                                                                                                                        
     & 100  & 2.80 & 1.70 & 70.04 & 102.65 & 2.16 & 0.72 & 70.07 & 88.11 & 2.00 & 0.47 & 75.62 & 98.99 \\                     
0.8  & 200  & 1.12 & 0.19 & 64.76 & 100.40 & 1.08 & 0.08 & 88.80 & 116.40 & 1.04 & 0.03 & 92.68 & 108.22 \\                   
     & 500  & 1.00 & 0.00 & 67.82 & 105.03 & 1.00 & 0.00 & 80.83 & 132.12 & 1.00 & 0.00 & 80.54 & 107.80 \\                   
     &1000  & 1.00 & 0.00 & 80.70 & 127.66 & 1.00 & 0.00 & 71.50 & 99.24 & 1.00 & 0.00 & 97.54 & 111.18 \\                    

\bottomrule
\end{tabular}
\end{table}

Table \ref{Tab4} shows the in-control process based on the ARMA(1,1) model and  disturbances included in the intercept parameter $\phi_{0}$, in which represent levels of change in the process mean, i.e., those time series are generated from different mean drifts. The $t_{\eb}$ is well suitable to capture this kind of change, as expected. Even being considered as the \textit{underdog} in this scenario, we noticed the $\mathcal{T}_{\betab}^2$ performance close to the residual one for the highest number of reference batches with the highest time-instants in any disturbances. In other words, even for the changes in the mean, instead of in the data dynamic, our approach seems to work well. 

\begin{table}[!ht] 
\caption{ARMA(1,1): Mean  $(\hat{\mu})$ and standard deviation  $(\hat{\sigma})$ of
 $ARL_0$ ($\phi_{0}=1$) and $ARL_1$ values for disturbances in in the intercept $\phi_{0}$}
 \label{Tab4}
\scriptsize
\centering
\begin{tabular}{c|c|cccc|cccc|cccc}
  \toprule
      \multirow{3}{*}{$\phi_0$} & \multirow{3}{*}{$n$} & \multicolumn{12}{c}{I} \\
\cline{3-14}
    &&\multicolumn{4}{c|}{10} &
      \multicolumn{4}{c|}{30 } &
      \multicolumn{4}{c}{100 } \\

\cline{3-14}
  &  & $\mathcal{T}_{\betab}^2(\hat{\mu})$ & $\mathcal{T}_{\betab}^2(\hat{\sigma})$ & $t_{\eb}(\hat{\mu})$&
  $t_{\eb}(\hat{\sigma})$&$\mathcal{T}_{\betab}^2(\hat{\mu})$ & $\mathcal{T}_{\betab}^2(\hat{\sigma})$ & $t_{\eb}(\hat{\mu})$& $t_{\eb}(\hat{\sigma})$& $\mathcal{T}_{\betab}^2(\hat{\mu})$ & $\mathcal{T}_{\betab}^2(\hat{\sigma})$ & $t_{\eb}(\hat{\mu})$& $t_{\eb}(\hat{\sigma})$\\
  \hline   
     & 100 & 1.07 & 0.08 & 1.00 & 0.00 & 1.04 & 0.03 & 1.00 & 0.00 & 1.02 & 0.01 & 1.00 & 0.00 \\                           
0.0  & 200 & 1.00 & 0.00 & 1.00 & 0.00 & 1.00 & 0.00 & 1.00 & 0.00 & 1.00 & 0.00 & 1.00 & 0.00 \\                           
     & 500 & 1.00 & 0.00 & 1.00 & 0.00 & 1.00 & 0.00 & 1.00 & 0.00 & 1.00 & 0.00 & 1.00 & 0.00 \\                           
     &1000 & 1.00 & 0.00 & 1.00 & 0.00 & 1.00 & 0.00 & 1.00 & 0.00 & 1.00 & 0.00 & 1.00 & 0.00 \\                           
\hline                                                                                                                      
     & 100 & 4.93 & 3.12 & 1.00 & 0.00 & 3.89 & 1.63 & 1.00 & 0.00 & 3.12 & 0.76 & 1.00 & 0.00 \\                           
0.5  & 200 & 1.91 & 0.74 & 1.00 & 0.00 & 1.52 & 0.32 & 1.00 & 0.00 & 1.44 & 0.23 & 1.00 & 0.00 \\                           
     & 500 & 1.02 & 0.03 & 1.00 & 0.00 & 1.01 & 0.01 & 1.00 & 0.00 & 1.00 & 0.00 & 1.00 & 0.00 \\                           
     &1000 & 1.00 & 0.00 & 1.00 & 0.00 & 1.00 & 0.00 & 1.00 & 0.00 & 1.00 & 0.00 & 1.00 & 0.00 \\                           
\hline                                                                                                                      
     & 100 & 63.39 & 77.90 & 3.85 & 8.47 & 38.90 & 32.02 & 2.41 & 1.86 & 30.80 & 16.68 & 2.30 & 1.62 \\                     
 0.8 & 200 & 40.05 & 74.13 & 1.26 & 0.94 & 20.85 & 16.44 & 1.14 & 0.17 & 16.60 & 8.25 & 1.13 & 0.13 \\                      
     & 500 & 8.68 & 7.46 & 1.01 & 0.04 & 5.22 & 2.15 & 1.01 & 0.01 & 4.63 & 1.51 & 1.00 & 0.00 \\                           
     &1000 & 2.61 & 1.48 & 1.00 & 0.00 & 1.93 & 0.55 & 1.00 & 0.00 & 1.75 & 0.29 & 1.00 & 0.00 \\                           
\hline                                                                                                                      
     & 100 & 163.85 & 158.72 & 143.12 & 165.64 & 107.67 & 90.98 & 156.41 & 171.88 & 94.23 & 78.88 & 176.62 & 159.59 \\      
 1.0 & 200 & 141.20 & 116.93 & 91.45 & 135.41 & 141.52 & 118.74 & 158.14 & 157.23 & 116.36 & 89.73 & 208.97 & 178.09 \\     
     & 500 & 174.39 & 151.85 & 125.42 & 158.62 & 154.68 & 135.68 & 161.39 & 178.08 & 120.63 & 90.08 & 207.97 & 181.85 \\    
     &1000 & 175.58 & 158.47 & 118.26 & 158.62 & 164.63 & 129.10 & 167.33 & 167.43 & 138.47 & 105.63 & 207.66 & 180.26 \\   
\hline                                                                                                                      
     & 100 & 66.71 & 91.16 & 2.81 & 3.92 & 41.20 & 34.88 & 2.00 & 2.42 & 31.48 & 14.39 & 1.52 & 0.45 \\                     
1.2  & 200 & 44.10 & 76.60 & 1.13 & 0.27 & 22.69 & 19.12 & 1.08 & 0.09 & 15.15 & 6.78 & 1.08 & 0.11 \\                      
     & 500 & 7.20 & 7.47 & 1.00 & 0.00 & 5.25 & 2.64 & 1.00 & 0.00 & 4.53 & 1.66 & 1.00 & 0.00 \\                           
     &1000 & 2.55 & 1.51 & 1.00 & 0.00 & 1.94 & 0.52 & 1.00 & 0.00 & 1.76 & 0.33 & 1.00 & 0.00 \\                           
\hline                                                                                                                      
     & 100 & 5.81 & 4.59 & 1.00 & 0.00 & 4.11 & 1.67 & 1.00 & 0.00 & 3.33 & 0.98 & 1.00 & 0.00 \\                           
1.5  & 200 & 1.83 & 0.86 & 1.00 & 0.00 & 1.62 & 0.44 & 1.00 & 0.00 & 1.38 & 0.19 & 1.00 & 0.00 \\                           
     & 500 & 1.02 & 0.03 & 1.00 & 0.00 & 1.01 & 0.01 & 1.00 & 0.00 & 1.00 & 0.00 & 1.00 & 0.00 \\                           
     &1000 & 1.00 & 0.00 & 1.00 & 0.00 & 1.00 & 0.00 & 1.00 & 0.00 & 1.00 & 0.00 & 1.00 & 0.00 \\                           
\hline                                                                                                                      
     & 100 & 1.07 & 0.06 & 1.00 & 0.00 & 1.03 & 0.02 & 1.00 & 0.00 & 1.02 & 0.01 & 1.00 & 0.00 \\                           
2.0  & 200 & 1.00 & 0.00 & 1.00 & 0.00 & 1.00 & 0.00 & 1.00 & 0.00 & 1.00 & 0.00 & 1.00 & 0.00 \\                           
     & 500 & 1.00 & 0.00 & 1.00 & 0.00 & 1.00 & 0.00 & 1.00 & 0.00 & 1.00 & 0.00 & 1.00 & 0.00 \\                           
     &1000 & 1.00 & 0.00 & 1.00 & 0.00 & 1.00 & 0.00 & 1.00 & 0.00 & 1.00 & 0.00 & 1.00 & 0.00 \\  
     \bottomrule
\end{tabular}
\end{table}

Tables S1 to S4 in the Supplementary Material show the results of Monte Carlo Simulations for ARMA(2,2), AR(1) and MA(1), respectively. They are very similar compared to the study presented in Tables \ref{Tab2} to \ref{Tab4}. 
\clearpage
\pagebreak

\subsection{Monitoring and diagnosis case}

This Section is aimed to show the full potential of our approach. For a simulated ongoing batch process we display the group of $T^2_{\beta}$ chart and $t_{\beta}$ charts for the individual coefficients in order to show its performance to detect and diagnose the imposed disturbances.  

Let's consider the industrial process generating batches of 200 time-length following the ARMA(1,1) model, with $\phi_0$, $\phi_1$ and $\theta_1$ set as in Table \ref{tab:tbnew}. We built the charts using $\alpha=0.01$ and considered 30 in-control batches as the reference. We simulate 20 new batches with two levels of disturbances in the $\phi_1$ parameter: (i) a moderate level (from $\phi_1$=0.2 to $\phi_1$=0); and (ii) a intense level (from $\phi_1$=0.2 to $\phi_1$=0.6).

Tables \ref{fig:phizero} and \ref{fig:phizero6} show the group of control charts for the moderate and intense level of disturbance, respectively. In Table \ref{fig:phizero} we noticed that the multivariate $T^2_{\beta}$ starts to signalize a number of points out of the limits just after the moderate disturbance has imposed (after the $30^{th}$ batch). It reinforces a good performance in the simulated study shown in Table \ref{Tab2}, as expected. Additionally, from the $t_{\beta}$ charts there is a good tip about the source of the disturbance, since the $t_{\beta}$ chart for $\phi_1$ points out some points above the limits. The other two $t_{\beta}$ charts remain with nearly all points randomly running within the charts boundaries, as it should be, since there are no disturbances imposed in $\phi_0$ and $\theta_1$. 

In Table \ref{fig:phizero6} it becomes even more pronounced since the level of disturbance is higher. We can see nearly all points out of limits in the $T^2_{\beta}$ and $t_{\beta}$ chart for $\phi_1$. The $t_{\beta}$ for $\phi_0$ show a few points outside the limits after the $30^{th}$ batch. Those few false alarms are likely to be due to coefficient covariance. In general, these charts seem to work really well to signalize and diagnose the source of disturbance.

\begin{figure}[!h]
\begin{center}
\includegraphics[scale=0.4]{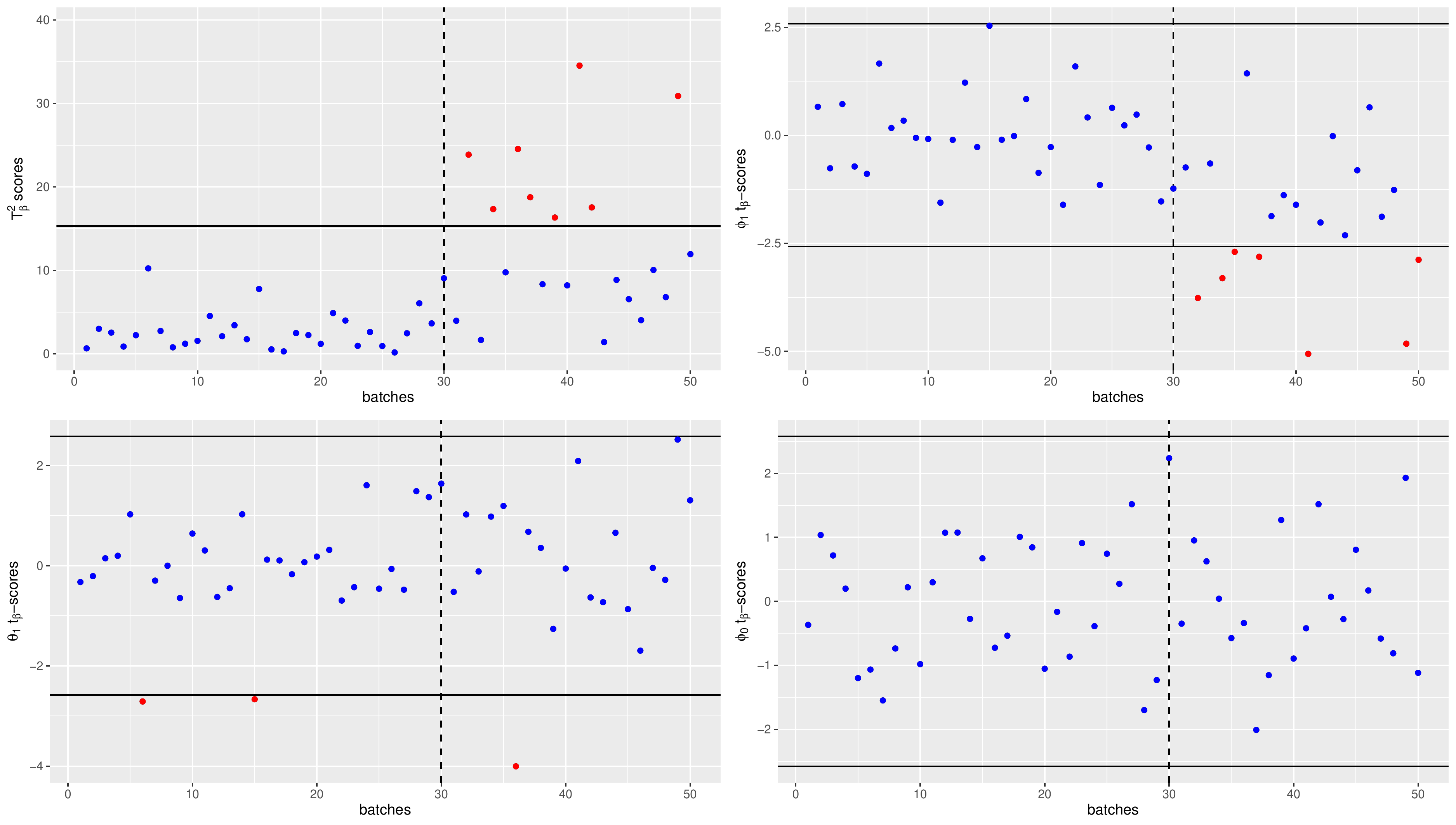}
\end{center} \vspace{-2mm}
\caption{$T^2_{\beta}$ and $t_{\beta}$ for 30 in-control batches and 20 new batches with a change in $\phi_{1}$ from 0.2 to 0.}
\label{fig:phizero}
\end{figure}

\begin{figure}[!h]
\begin{center}
\includegraphics[scale=0.4]{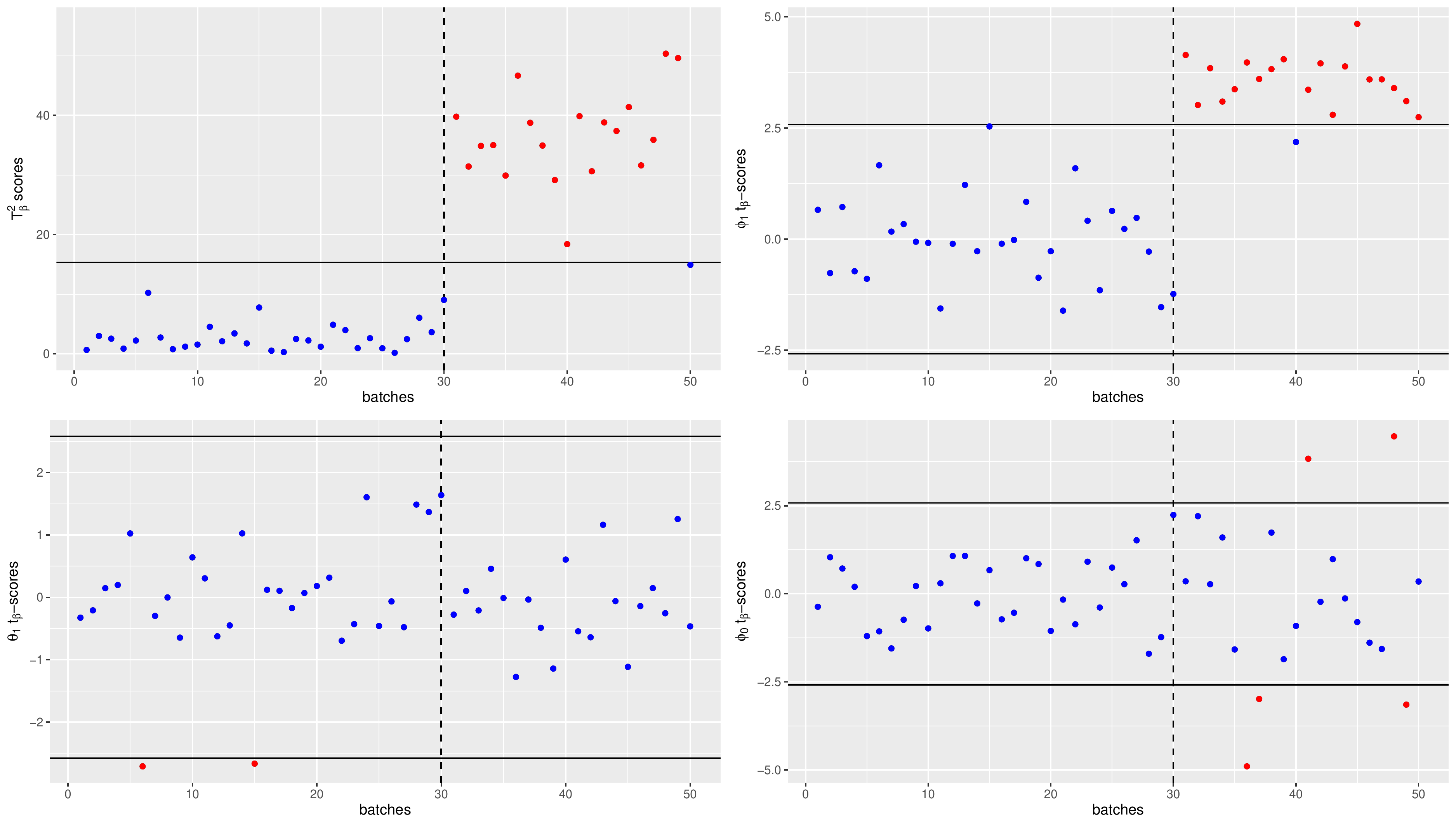}
\end{center} \vspace{-2mm}
\caption{$T^2_{\beta}$ and $t_{\beta}$ for 30 in-control batches and 20 new batches with a change in $\phi_{1}$ from 0.2 to 0.6.
}
\label{fig:phizero6}
\end{figure}

\clearpage
%\pagebreak

\section{Application}

In order to illustrate the applicability of our methodology we consider a real dataset of time series representing measurements of engine noise from \cite{TimeSeriesClass}. We can understand each time series as one batch sampled from an industrial process. Let's consider only the training dataset which has 3271 batches, each one with 500 time instants, sampled from the process operating under two different conditions, labeled as +1 and -1 with 1755 and 1846 number of batches, respectively. In this application we assume that the group labeled as +1 is the reference group, i.e., sample batches coming from the process operating in a standard condition.

The chart on the left on Figure (\ref{fig:acfar}) shows the average behaviour of the autocorrelation function (ACF) for all batches according to their group. The colours green and red represent the reference group (labeled as +1) and the monitoring group (labeled as -1), respectively. We note that the main feature is the cyclical behavior of data in both groups. Following the aim of our methodology, in Phase $\mathcal{I}$ data from the reference group are modeled by using the class of ARMA models. We know that this class of models is suitable to capture the time series dynamic rather than other features like cycles, trends, etc. For that reason we chose an AR model of high order to capture the cycles as the dynamics. Here an AR(12) was adopted in order to model this feature since the order 12 was necessary and sufficient for getting uncorrelated and normally distributed residuals. The Ljung–Box test pointed out that 100\% of the residuals are uncorrelated, non significant to presence of correlation and 95\% of the residuals were normally distributed according to the Shapiro-Wilk test for normality, using a significance level of 5\%. 
 
Although we have 12 AR coefficients in the model for the reference group, we noticed in Figure (\ref{fig:acfar}) on the right chart that the large portion of them is possibly not significantly different from zero. Thus, proceeding an individual t-test for each  coefficient and verifying their statistical significance at a level of 5\%, only the first three coefficients of the fitted AR model were significant in 95\% of batches. These results are summarized in Table \ref{tab:sig}. For this reason, the $\mathcal{T}_{\betab}^{2}$ control chart was built with the first three coefficients. Its seems to be a good choice insofar as Figure (\ref{fig:acfar}) shows a clear visual difference between the means of the adjusted coefficients from group +1 and -1 just in those 3 first AR coefficients (represented by green and red lines).

\begin{table}[h!]
\caption{Rate of significant coefficients in the individual t-test for the reference group adjusted AR(12) model.\linebreak}\label{tab:sig}
    \centering
    \begin{tabular}{c|c|c|c|c|c|c|c|c|c|c|c|c}
    \toprule
     \# Batches    & $\phi_1$ & $\phi_2$ & $\phi_3$ & $\phi_4$ & $\phi_5$ & $\phi_6$ & $\phi_7$ & $\phi_8$ & $\phi_9$ & $\phi_{10}$ & $\phi_{11}$ & $\phi_{12}$  \\
         \hline 
   1755    & 1.00 & 1.00 & 0.98 &0.35 &0.48& 0.20 &0.30& 0.20 &0.18 &0.25 &0.29 &0.51\\
        \bottomrule
   \end{tabular}
\end{table}

\begin{figure}[!h]
\begin{center}
\includegraphics[scale=0.4]{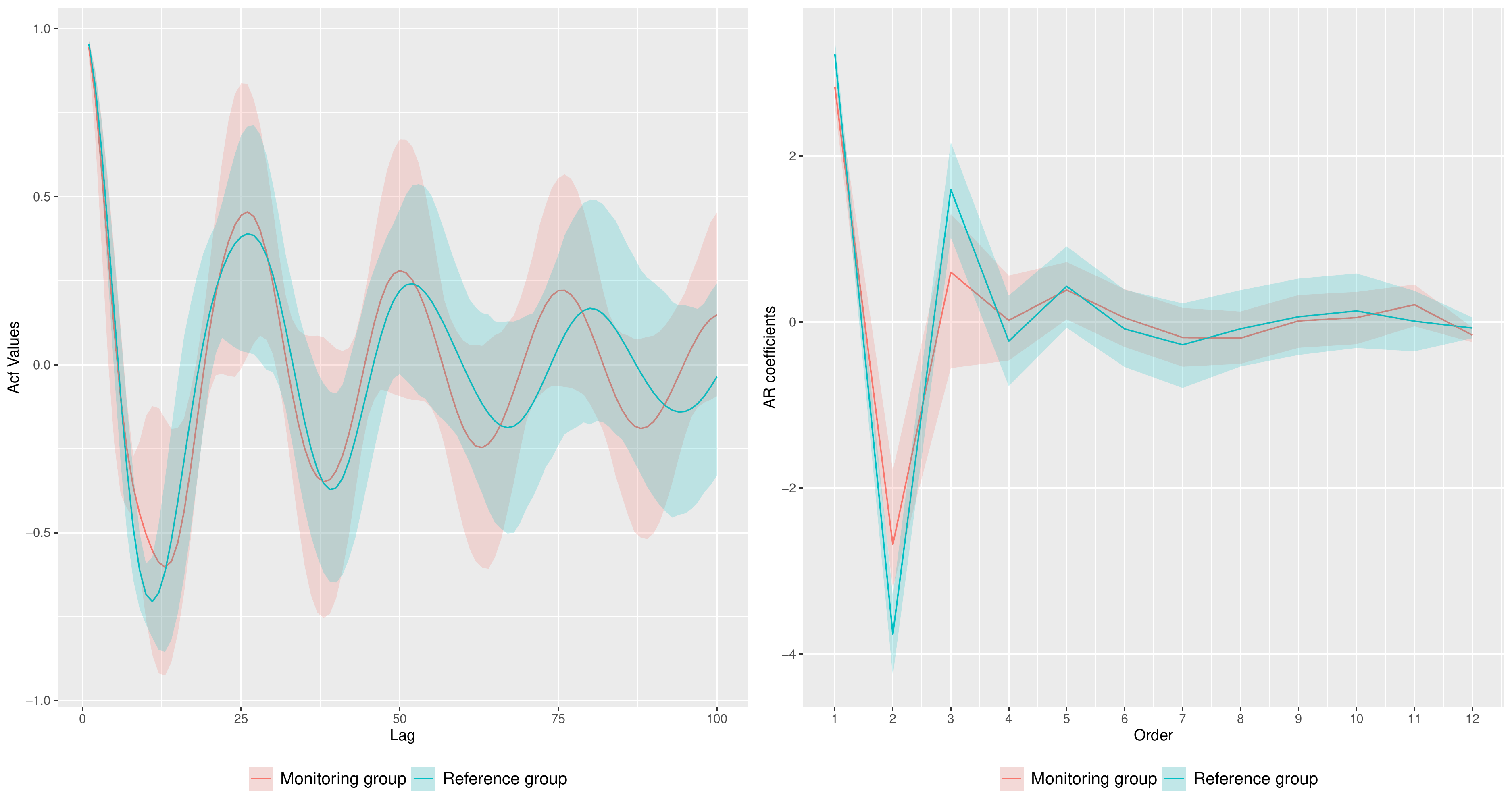}
\end{center} \vspace{-2mm}
\caption{Autocorrelations for the FordA-train dataset and AR coeficients from the ajusted  AR(12) model for all batches.
}
\label{fig:acfar}
\end{figure}

In order to show the performance of our approach we take randomly $I$ in-control batches (i.e., from the ones labeled as +1) as the reference batches to fit the AR(12) model and build the $\mathcal{T}_{\betab}^2$ chart. We do variate the study by using $I$ values of 10, 30, 100, 200, 300, 500 with 200 replications each. The control limits are set with false alarm probability ($\alpha$) of 0.10, 0.05 and 0.01. The remaining 1755 - $I$ batches are used to evaluate the empirical false alarm probability. The 1846 batches labeled as -1 are used to evaluate the power of $\mathcal{T}_{\betab}^2$ chart.

\begin{table}[ht]
\caption{AR(12): Mean $(\hat{\mu})$ and standard deviation $(\hat{\sigma})$  of
 $r_0$ (Phase $\mathcal{I}$) and $r_1$ (Phase $\mathcal{II}$) values}
\label{machineL}
\centering
\begin{tabular}{c|c|cc|cc|cc}
\toprule
&\multicolumn{7}{c}{False alarm probability ($\alpha$)} \\
\hline
  &  &\multicolumn{2}{c|}{0.10} &
       \multicolumn{2}{c|}{0.05} &
       \multicolumn{2}{c}{0.01} \\

\cline{3-8}

  & $I$ & $\mathcal{T}_{\betab}^2(\hat{\mu})$  & $\mathcal{T}_{\betab}^2(\hat{\sigma})$  & $\mathcal{T}_{\betab}^2(\hat{\mu}) $ & $\mathcal{T}_{\betab}^2(\hat{\sigma})$& $\mathcal{T}_{\betab}^2 (\hat{\mu})$ & $\mathcal{T}_{\betab}^2(\hat{\sigma})$\\ 
  \hline
 \multirow{5}{*}{\STAB{\rotatebox[origin=c]{90}{Phase $\mathcal{I}$}}} &  10  & 0.13 & 0.10  &  0.08 & 0.08  & 0.03 & 0.04  \\ 
  & 100 & 0.10 & 0.02  & 0.07 & 0.02  & 0.03 & 0.01  \\ 
   & 200 & 0.10 & 0.02  & 0.07 & 0.01  & 0.03 & 0.01  \\ 
    & 300 & 0.09 & 0.01  & 0.06 & 0.01  & 0.03 & 0.01  \\ 
 & 500 & 0.08 & 0.01  & 0.05 & 0.01  & 0.02 & 0.01  \\ 
\hline 
\multirow{5}{*}{\STAB{\rotatebox[origin=c]{90}{Phase $\mathcal{II}$}}}  & 10  & 0.76 & 0.17  & 0.61 & 0.22  & 0.31 & 0.22  \\ 
 & 30  & 0.87 & 0.07  & 0.79 & 0.09  & 0.57 & 0.15  \\ 
 & 100 & 0.89 & 0.03  & 0.83 & 0.04  & 0.67 & 0.07  \\ 
 & 200 & 0.89 & 0.02  & 0.83 & 0.03  & 0.68 & 0.04  \\ 
&  300 & 0.89 & 0.02  & 0.84 & 0.02  & 0.69 & 0.03  \\ 
 & 500 & 0.89 & 0.01  & 0.84 & 0.02  & 0.70 & 0.03  \\ 
   \hline
\end{tabular} 
\end{table}

Table \ref{machineL} summarize the results of $\mathcal{T}_{\betab}^2$ chart. The $r_0$ and $r_1$ are the rate of false alarm and disturbed batches detected, respectively. We noticed that the observed $r_{0}$ (highlighted in the gray line) is closer to the chosen nominal value of $\alpha$ as the number of samples increases, which is consistent with the theoretical distribution derived in Theorem \ref{theo2}. The $r_{1}$ values show the performance of $\mathcal{T}_{\betab}^2$ chart to signalize the out-of-control (labeled as -1) batches. As we expected the degree of detection increases as the number of reference batches $I$ in phase $\mathcal{I}$ increases. Even for the very small number of batches compared to the overall number of batches available labeled as +1, the $\mathcal{T}_{\betab}^2$ shows a good rate of detection for each false alarm probability $\alpha$.

\clearpage
\pagebreak

\section{Conclusion}

This paper introduced a new approach to deal with batch processes through a set of ARMA-based control charts. Through in-control batch samples available we fitted the ARMA model and built a group of charts based on the coefficient estimates from historical in-control batches. The modified \textit{Hotelling} and \textit{t-Stutent} distributions can easily accommodate those estimates and a decision rule was made for monitoring future samples. Additionally, the modified t-Student charts help to look for the source of disturbances.

The simulated batch process generating samples of time series from an ARMA model was presented. The $T^2_{beta}$ chart outperforms the traditional competitor based on the residuals for detecting changes of any level in the process dynamic.  Furthermore, we have shown how powerful are the individual $t_{\beta}$ charts to identify the source of disturbances imposed in the process. 

The applicability of our approach was illustrated through a real data set in which the good performance is clearly noticed, even for a very small sample reference batches compared to the overall number of in-control batches available. 

Finally, it's important to noticed that we can built the group of $T^2_{beta}$ and $t_{\beta}$ charts from any ARMA ($v$,$w$) sub model, including only AR($v$) or MA($w$) component with the order less or equal to $v$ or $w$, with and without intercept. It opens the applicability of this approach to a wide range of batch processes. 

\section*{Declaration of Competing Interest}

The authors declare that they have no known competing financial interests or personal relationships that could have appeared to influence the work reported in this paper.

\section*{Supplementary material}

\begin{description}
\item[Supplementary material:] Supplementary tables.
\item[Code:] R-functions containing all methods developed in this article (will be available in the dvqcc package at CRAN).
\item[Data:] Dataset used in the application and corresponding script (zip).
\end{description}

\bibliography{references}

\end{document}